\documentclass[runningheads]{llncs}
\usepackage[T1]{fontenc}
\usepackage{graphicx}
\usepackage{booktabs}
\usepackage[numbers]{natbib}
\usepackage[misc]{ifsym}
\newcommand{\corr}{(\Letter)}
\usepackage{mwe}
\usepackage{amssymb}
\usepackage[normalem]{ulem}

\usepackage{amsmath}
\usepackage{algorithm}
\usepackage[noend]{algpseudocode}

\usepackage{booktabs}   % for \toprule, \midrule, \bottomrule
\usepackage{multirow}   % for multi-row cells
\usepackage{array} 
\usepackage{siunitx}
\usepackage{subcaption}
\usepackage{adjustbox}
\usepackage{hyperref}
\DeclareMathOperator*{\maximize}{maximize}
\DeclareMathOperator{\argmax}{argmax}
\DeclareMathOperator{\argmin}{argmin}

\usepackage{mathtools}

\newcommand{\set}[1]{\left\{#1\right\}}
\newcommand{\pr}[1]{\left(#1\right)}
\newcommand{\fpr}[1]{\mathopen{}\left(#1\right)}

\newcommand{\abs}[1]{{\left|#1\right|}}

\newcommand{\ceil}[1]{{\left\lceil#1\right\rceil}}

\newcommand{\np}{\textbf{NP}}

\newcommand{\define}{\leftarrow}

\DeclareRobustCommand{\dispfunc}[2]{%
	\ensuremath{%
		\ifthenelse{\equal{#2}{}}%
			{\mathit{#1}}%
			{\mathit{#1}\fpr{#2}}}}

\newcommand{\bigO}[1]{\dispfunc{\mathcal{O}}{#1}}

\newcommand{\score}{f}
\newcommand{\dens}{\Gamma}

\newcommand{\gdsh}{\textsc{swamp}\xspace}
\newcommand{\ugdsh}{\textsc{swamp}\xspace}
\newcommand{\dsg}{\textsc{dsg}}
\newcommand{\dss}{\textsc{dss}}
\newcommand{\card}{\textsc{card-swamp}}
\newcommand{\fair}{\textsc{fair-swamp}}

\newcommand{\contract}{\textsc{contract}\xspace}

\newcommand{\exact}{\textsf{Exact}}
\newcommand{\exactflow}{\textsf{ExactFlow}}
\newcommand{\peelmax}{\textsf{PeelMax}}
\newcommand{\greedypeel}{\textsf{Greedy}}
\newcommand{\peelzero}{\textsf{PeelZero}}
\newcommand{\degreepeel}{\textsf{DegPeel}}

\newcommand{\projflow}{\textsf{ProjFlow}}

\newcommand{\projgreedy}{\textsf{ProjGreedy}}

\usepackage{enumitem}
\renewcommand{\refname}{References}
\makeatletter
\renewcommand\bibsection{%
  \section*{{\refname}\@mkboth{\refname}{\refname}}%
}%
\makeatother

\tocauthor{Vedangi Bengali, Nikolaj Tatti, Iiro Kumpulainen, Florian Adriaens, Nate Veldt}
\toctitle{The Densest SWAMP problem: subhypergraphs with arbitrary monotonic partial edge rewards}
\begin{document}

\title{The Densest SWAMP problem: subhypergraphs with arbitrary monotonic partial edge rewards}

\titlerunning{The Densest SWAMP Problem}
% If the full title of your paper is short enough to also fit in the running head, you can omit the abbreviated paper title here. You can check as follows: if you comment out the \titlerunning line, something will appear in the header of all odd-numbered pages of your PDF from page 3 onward. This something is either the full title (in which case all is well), or the error message "Title Suppressed Due to Excessive Length". If this error message appears, you're going to want to provide an abbreviated title within the \titlerunning command, because if you won't do it, Springer will do it for you.

%N.B.: Author information (both in the \author{} and \authorrunning{} command) should only be present in the Camera-Ready Version of your paper. The version that you initially submit for review, ought to be double-blind. So, when initially submitting your paper, use:

% \author{Author information scrubbed for double-blind reviewing}

% %%%%%%%%%%%%%%%%%%%%%%%%%%%%%%%%%%%%%%%%%%%%%%%%
\author{Vedangi Bengali\inst{1} \corr \and
Nikolaj Tatti\inst{2} \and
Iiro Kumpulainen\inst{2} \and
Florian Adriaens\inst{2} \and
Nate Veldt\inst{1}}  % \orcidID{0000-1111-2222-3333}
%%%%%%%%%%%%%%%%%%%%%%%%%%%%%%%%%%%%%%%%%%%%%%%%%%

% You may leave out the orcidID information, if you want to.
% Use \corr to indicate the corresponding author. Note the spacing around the \corr command. Only one author can be the corresponding author.

%N.B.: comment out the \authorrunning{} command for the double-blind version of your paper submitted for review. Later, if your paper is accepted, use the command for the Camera-Ready Version.
%%%%%%%%%%%%%%%%%%%%%%%%%%%%%%%%%%%%%%%%%%%%%%%%%%
\authorrunning{V. Bengali et al.}
%%%%%%%%%%%%%%%%%%%%%%%%%%%%%%%%%%%%%%%%%%%%%%%%%%

% If there are two authors, write 'A.L. Benjamin and C.C. Broadus Jr.'
% If there are more than two authors, '[...] et al.' is used.
\institute{}
%%%%%%%%%%%%%%%%%%%%%%%%%%%%%%%%%%%%%%%%%%%%%%%%%%
\institute{Texas A\&M University, College Station, TX, USA\\ \email{\{vedangibengali, nveldt\}@tamu.edu}
\and
University of Helsinki, HIIT, Helsinki, Finland\\
\email{\{nikolaj.tatti, iiro.kumpulainen, florian.adriaens\}@helsinki.fi}}
% \and
% Address3
% \email{ @ .edu}}
%%%%%%%%%%%%%%%%%%%%%%%%%%%%%%%%%%%%%%%%%%%%%%%%%%

\maketitle              % typeset the header of the contribution

\begin{abstract}
We consider a generalization of the densest subhypergraph problem where nonnegative rewards are given for including partial hyperedges in a dense subhypergraph. 
Prior work addressed this problem only in cases where reward functions are convex, in which case the problem is poly-time solvable. We consider a broader setting where rewards are monotonic but otherwise arbitrary. We first prove hardness results for a wide class of non-convex rewards, then design a $1/k$-approximation by projecting to the nearest set of convex rewards, where $k$ is the maximum hyperedge size. We also design another $1/k$-approximation using a faster peeling algorithm, which (somewhat surprisingly) differs from the standard greedy peeling algorithm used to approximate other variants of the densest subgraph problem. Our results include an empirical analysis of our algorithm across several real-world hypergraphs.

\end{abstract}

\section{Introduction}
Dense subgraph discovery is a widely studied primitive in graph mining, with applications including team formation~\cite{gajewar2012multi,rangapuram2013towards}, motif discovery~\cite{fratkin2006motifcut}, and fraud detection~\cite{chen2022antibenford}. One of the most common problems in dense subgraph discovery is the densest subgraph problem (\dsg{}). For a graph $G = (V,E)$, \dsg{} seeks a node set $S \subseteq V$ that maximizes the ratio between (induced) edges and nodes, i.e.,
\begin{equation}
\label{eq:dsg}
\maximize_{S \subseteq V} \quad \frac{|E(S)|}{|S|},
\end{equation}
where $E(S) = \{ (u,v) \in E \colon u,v \in S\}$. There are known polynomial-time algorithms for exactly solving \dsg{} using flow-based methods~\cite{goldberg1984finding} and linear programming~\cite{charikar2000greedy}. There is also a simple greedy $1/2$-approximation based on \emph{peeling} (iteratively removing a minimum degree node)~\cite{asahiro2000greedily,charikar2000greedy}. Many results for variants of \dsg{} focus on extending one or more of these basic techniques (flow, linear programming, and peeling) to more general settings~\cite{khuller2009finding,chekuri2022densest,veldt2021generalized,lanciano2024survey,tsourakakis2015k,huang1995clusters,huang2024densest,zhou2022extracting}.

This paper focuses on generalizations of \dsg{} to hypergraphs, which group nodes into (hyper)edges involving an arbitrary number of nodes (rather than just two).
As a motivating application that we will consider throughout the manuscript, finding dense regions of a hypergraph provides a particularly natural approach for team formation. There are already a number of methods for team formation based on dense subgraph discovery~\cite{gajewar2012multi,rangapuram2013towards,miyauchi2023densest}. These operate under the assumption that a dense region of a social network represents a good team for a future task, since it encodes a group of people who have already collaborated extensively in the past. Using a hypergraph rather than a graph for this application is arguably more natural, since it directly captures entire previous \emph{team} interactions as hyperedges, rather than only previous pairwise relationships. 

The simplest generalization of \dsg{} to hypergraphs is to find a node set $S$ that maximizes the ratio between the number of edges \emph{completely contained in $S$} and $|S|$. This was first considered 30 years ago in the context of circuit decomposition~\cite{huang1995clusters} and has been considered by many others since~\cite{hu2017maintaining,bera2022new}. While this is a natural extension of the standard \dsg{} problem in graphs, there are many applications in which \emph{partially} including a hyperedge also intuitively contributes to notions of density. Consider again the example of team formation. 
When forming a new team, adding even a \emph{subset} of people from a previous collaboration intuitively matches the belief that previous collaborations help contribute to good future team interactions. However, the standard densest subhypergraph objective only gives a reward for including an entire hyperedge (i.e., all members of some previous team) to the new team. 

In this paper, we focus on new algorithms for a generalized densest subhypergraph objective introduced by \citet{zhou2022extracting}, which incorporates positive rewards for partially included edges. In more detail, each edge has a nonnegative monotonic reward function $r$, and the overall objective is to maximize a sum of partial rewards divided by the size of the output set (Section~\ref{sec:swamp} includes a formal definition). \citet{zhou2022extracting} proved that there exist reward functions for which this general problem becomes \np-hard. However, they then focused on convex reward functions, in which case the problem is a special case of the poly-time solvable densest supermodular subset problem~\cite{chekuri20221}. They designed an exact algorithm for this convex case based on solving maximum $s$-$t$ flow problems, and showed that a standard greedy peeling algorithm yields a $1/k$-approximation where $k$ is the maximum hyperedge size.

\textbf{Our contributions.} We focus on the problem of finding a densest Subhypergraph With Arbitrary Partial edge rewards (\gdsh{}). By arbitrary, we mean that rewards are not required to be convex, as was the focus for~\citet{zhou2022extracting}. Our contributions include the following:
\begin{itemize}
    \item
    We completely settle the complexity of \gdsh{} for all reward functions, significantly strengthening prior hardness results. Concretely, when applying the same reward function $r$ to every edge, \gdsh{} is either poly-time solvable because $r$ is convex, trivially optimized by a single node if $r$ is non-convex but satisfies a certain extremal condition, or is otherwise \np-hard. 
    
    \item
    We design peeling algorithms for the non-convex case that come with a $1/k$-approximation. Somewhat surprisingly, this approximation is not obtained by the standard greedy peeling method (which gives a $1/k$ approximation for the convex case), but rather with a peeling method that makes locally suboptimal choices about which node to remove in each step. 
    
    \item
    We design approximation algorithms based on projecting non-convex rewards to convex rewards and then solving the latter problem. In the worst case, these come with a $1/k$-approximation.

    \item
    We introduce a new integer linear programming formulation for finding optimal solutions to \gdsh{}, that works even for (small-scale) \np-hard cases. 

    \item
    We implement these algorithms and demonstrate their performance on real-world datasets, showing they exceed their theoretical guarantees in practice.
    
    \item
    We show that as a consequence of having approximation guarantees for \gdsh{}, we obtain approximations for constrained variants of our objective. 
    
\end{itemize}

\textbf{Other related work.} We refer to \citet{lanciano2024survey} for an extensive recent survey on variants of the densest subgraph problem. We briefly cover several results that are particularly relevant to our paper.  Several variants of \dsg{} include constraints on the number of nodes in the output set $S$, including two variants introduced by \citet{andersen2009finding} called the densest at-most-$k$ ($|S| \leq k$), and the densest at-least-$k$ ($|S| \geq k$) problems. In other problem variants, the goal is to find a dense subgraph that ensures nodes with different node labels are included in the output. In some cases, node labels represent ``skills'' and the goal is to form a (dense) team of individuals to cover certain skill sets~\cite{gajewar2012multi,rangapuram2013towards}. In other cases, node labels represent disjoint protected classes and the goal is to form a dense subgraph that is ``diverse'' or ``fair''~\cite{anagnostopoulos2020spectral,miyauchi2023densest,kariotakis2025fairness}. Approximation algorithms for several of these variants build upon earlier approximation algorithms for the densest at-least-$k$ variant~\cite{gajewar2012multi,miyauchi2023densest}.

Recently, Chekuri et al.~\cite{chekuri2022densest} introduced the densest supermodular subset problem (\dss{}) where the goal is to maximize $F(S)/|S|$ where $F \colon 2^V \rightarrow \mathbb{R}^+$ is a nonnegative monotone supermodular function defined over a ground set $V$. Their results for this objective include a fast flow-based approximation algorithm, faster greedy peeling methods, and extensions to cardinality-constrained variants (e.g., the constraint $|S| \geq k$). The greedy peeling method starts with the entire node set $V$, finds a node $v = \argmin F(V) - F(V - v)$, and then removes it. At the next step, the same strategy is applied to the remaining node set. At the end, the method chooses the subset of nodes considered along the way with the best objective. This strategy is greedy in the sense that it removes the node that leads to the smallest decrease in the numerator of the objective at each step. This peeling method directly generalizes existing peeling methods for several special cases~\cite{charikar2000greedy,veldt2021generalized,zhou2022extracting}, including the $1/k$-approximation for the generalized densest subhypergraph objective of \citet{zhou2022extracting}.

 \section{The Densest SWAMP Problem}
\label{sec:swamp}
We consider a generalized maximum density problem with a hypergraph input $H = (V,E,w, \{r_e \colon e \in E\})$. Each $e \in E$ comes with a non-negative scalar weight $w(e) \geq 0$ and an \emph{edge-reward function} $r_e \colon \{0,1,\hdots,|e|\} \rightarrow \mathbb{R}_{\geq 0}$. The latter assigns a non-negative reward based on the \emph{number} of nodes in the hyperedge included in a node set $S$, even if $e$ is not entirely contained in $S$. We assume the reward function is monotonic and gives no reward for including no nodes from the edge:
\begin{equation*}
    0 = r_e(0) \leq r_e(1) \leq r_e(2) \leq \cdots \leq r_e(|e|).
\end{equation*}
This encodes the belief that including more of a hyperedge within a set should contribute more to the measure of density. 

\begin{problem}
    \label{pr:gdsh}
    For input $H = (V,E,w,\{r_e\colon e \in E\}$, the Densest Subhypergraph With Arbitrary Monotonic Partial edge rewards problem (\gdsh{}) seeks to solve
    \begin{align} \label{eq:gdsh}
    \max_{S \subseteq V}\quad \Gamma(S) = \frac{f(S)}{|S|}, \quad \text{ where }  f(S) = \sum_{e \in \mathcal{E}} w(e)\cdot r_e(|e \cap S|).
\end{align}
\end{problem}
If the rewards satisfy $r_e(i) = 0$ for $i < |e|$ and $r_e(|e|) = 1$, then this corresponds to the standard densest subhypergraph objective.
Zhou et al.~\cite{zhou2022extracting} were the first to study Problem~\ref{pr:gdsh}, though focused almost exclusively on the convex case. A discrete monotonic edge-reward function $r_e:\{0,1,\hdots,|e|\} \rightarrow \mathbb{R}_{\ge 0}$ is said to be convex if:
\begin{align}
    r_e(i+1)-r_e(i) \ge r_e(i) - r_e(i-1)  \quad \text{for } i \in \{1,2,\hdots, |e|-1\}.
\end{align}
We instead focus on the case where rewards are monotonic but otherwise {arbitrary}, i.e., not restricted to be convex.

In our study of hardness results, we consider a special case of \gdsh{} where the input is a hypergraph with maximum edge size $k$ and we use the same reward function $r$ for every edge. We refer to this as $\ugdsh(r)$. In this case, $r$ represents a parameter defining the \emph{problem} rather than a reward function that is inherently part of the input. Our goal is to characterize the complexity of $\ugdsh(r)$ under different choices for $r$.

We also consider two-sized constrained variants. The first is to maximize $\Gamma(S)$ subject to a cardinality constraint $|S| \geq \ell$ (where $\ell$ is an input to the problem), which we refer to as \card{}. We also consider a fair variant.
\begin{problem}
\label{pr:fair}
    Assume nodes of $H = (V,E,w, \{r_e\})$ are partitioned into node classes $\{C_1, C_2, \hdots, C_m\}$, and let a parameter $\ell_i$ be given for $C_i$ for $i = 1, 2, \hdots, m$. The Densest Fair SWAMP problem (\fair{}) is defined by
    \begin{equation}
        \maximize_{S \subseteq V} \quad \Gamma(S) \quad \text{subject to $|S\cap C_i| \geq \ell_i$, for $i = 1, 2, \hdots, m$}.
    \end{equation}
\end{problem}
If we think of nodes as individuals, the classes in Problem~\ref{pr:fair} can represent \emph{skills} that must be present when forming a team of individuals based on their past collaborations. As another example, classes could represent categories that must be fairly represented in the output set (e.g., selecting faculty members from all faculty ranks to represent an academic department on some internal committee).
 \section{Computational complexity of $\ugdsh(r)$}

The complexity of $\ugdsh(r)$ is known in the convex case.
\begin{theorem}[\citep{zhou2022extracting}]
\label{thr:poly}
If $r$ is convex, $\ugdsh(r)$ is polynomial-time solvable.
\end{theorem}
In addition, \citet{zhou2022extracting} showed that for \emph{some} choice of reward functions (where some rewards are convex but others are not), Problem~\ref{pr:gdsh} is \np-hard. We strengthen this result by showing that for \emph{every} non-convex $r$ (minus a corner case, described below), $\ugdsh(r)$ is \np-hard. This provides us with a full picture of the computational complexity of $\ugdsh(r)$.

Let us first describe a corner case when the optimal solution is a single node. This generalizes a result of \citet{zhou2022extracting} which says that if $r$ is concave, a single node defines the optimal solution.
\begin{theorem}
\label{thr:corner}
Assume $r$ such that for every $i \in \{1,2, \hdots, k\}$, it holds that $r(1)\geq r(i)/i$. Then there is an optimal solution for $\ugdsh(r)$ consisting of a single node.
\end{theorem}

\begin{proof}
For instance $H = (V, E, w)$, let $O$ be the solution for $\ugdsh(r)$. Let $o \in O$ be the node maximizing $\score(\set{o})$. Then
\begin{align*}
    \dens(O)
    & = \frac{\score(O)}{\abs{O}}
    = \frac{1}{\abs{O}}\sum_{e \in E} w(e)\cdot r(\abs{e \cap O}) \\
    & \leq \frac{1}{\abs{O}}\sum_{e \in E} w(e) \cdot r(1) \cdot \abs{e \cap O} = \frac{1}{\abs{O}}\sum_{x \in O} \score(\set{x}) \leq \score(\set{o}) = \dens(\set{o}),
\end{align*}
proving the claim.\qed
\end{proof}

The main result of this section now states that if the conditions in Theorems~\ref{thr:poly}--\ref{thr:corner} are not satisfied, then $\ugdsh(r)$ is \np-hard, see Appendix for proof.
\begin{theorem}
Assume $r$ such that for some positive integers $i$ and $j$ we have $r(i)/i > r(1)$ and $2r(j) > r(j-1) + r(j+1)$. 
Then $\ugdsh(r)$ is \np-hard.
\end{theorem}

 \section{Solving \gdsh with an integer linear program}\label{sec:ilp}
We introduce a new ILP formulation for solving the decision version of the problem, which asks whether there exists some set $S$ with density greater than $\alpha$ for some pre-specified $\alpha$. More formally, finding $S$ with $\dens(S) > \alpha$ is equivalent to finding $S$ for which $\score(S) - \alpha \abs{S} > 0$. We will solve the latter using an ILP, and once we have the solution, we can solve \gdsh by performing a binary search on $\alpha$. This construction is similar to the approach by~\citet{goldberg1984finding} for finding the densest subgraph.

We consider the following integer linear program.
\begin{align}
	\label{eq:maxecc}
\begin{aligned}
	\text{maximize} \quad& \left( \sum_{e \in E} w(e) \sum_{i = 1}^{|e|} \delta_{e,i} \cdot y_{e,i} \right) - \alpha \sum_{v \in V} x_v \hspace*{-1cm}\\ % ugly -1cm trick but such is life
	\text{s.t.} \quad& y_{e,i} \leq \frac{1}{i} \sum_{v \in e} x_v & \text{for } e \in E,\  i = 1,2, \hdots, |e|, \\
	&y_{e,i} \in \{0,1 \} & \text{for } e \in E,\  i = 1,2, \hdots, |e|, \\
	&x_v \in \{0,1 \}  & \text{for } v \in V.\\
\end{aligned}
\end{align}
In the above, we have defined
\begin{equation*}
    \delta_{e,i} = r_e(i) - r_e(i-1) \geq 0
\end{equation*}
so that if we include $k$ nodes from edge $e$ in the set $S$, we know that this gives a reward of
\begin{equation*}
 r_e(k) = \sum_{i = 1}^k \delta_{e,i}.
\end{equation*}
This ILP includes a variable $x_v \in \{0,1\}$ that indicates whether node $v$ is contained in the optimal density set $S$ ($x_v = 1$) or not ($x_v = 0$). The variable $y_{e,i}$ is designed to satisfy
\begin{equation}
\label{eq:yei}
y_{e,i} = \begin{cases}
1 & \text{ if $|e \cap S| \geq i$}\\
0 & \text{ otherwise.}
\end{cases}
\end{equation}
Observe that if this is the case, then the first part of the objective function of the ILP is exactly the edge reward
\begin{align*}
\sum_{e \in E} w(e) \sum_{i = 1}^{|e|} \delta_{e,i} \cdot y_{e,i} = \sum_{e \in E} w(e) \sum_{i = 1}^{|e\cap S|} \delta_{e,i} = \sum_{e \in E} w(e) \cdot r_e(|e \cap S|).
\end{align*}
We just need to confirm that for the optimal solution, the $y_{e,i}$ variables will indeed satisfy Eq.~\eqref{eq:yei}. Observe first of all that maximizing the objective, plus the fact that the $\delta_{e,i}$ are all positive, will ensure that the $y_{e,i}$ variables will be set to 1 whenever possible. Now, note that the constraint
\begin{equation*}
    y_{e,i} \leq \frac{1}{i} \sum_{v \in e} x_v
\end{equation*}
is equivalent to a bound $y_{e,i} < 1$ if $|e \cap S| < i$, and otherwise it amounts to a bound $y_{e,i} \leq c$ for some $c \geq 1$ if $|e \cap S| \geq i.$

 \section{Approximating \gdsh}
Our most significant algorithmic contributions are peeling-based approximation algorithms for the \np-hard regime of \gdsh{}. Our findings are surprising for two reasons. First of all, our peeling approach for the \np-hard non-convex case has a $1/k$-approximation, which is just as good as the peeling approximation guarantee for the (poly-time solvable) convex case. The second surprise is that our approximation guarantees \textit{do not} come from using the standard existing greedy peeling algorithm, which peels away a node in a way that leads to the best objective in the subsequent step. Rather, our guarantees work only for certain peeling strategies that may make locally suboptimal choices for which node to remove at each step. In addition to our peeling algorithms, we design another $1/k$-approximation based on projecting non-convex reward functions to convex rewards and then exactly solving the resulting convex problem.

\subsection{Peeling algorithm}
Here we describe a peeling algorithm, which will result in a $1/k$ approximation for \gdsh.
For notational simplicity, we will assume that $w(e) = 1$. Note that we can make this assumption without the loss of generality since we can incorporate the weights directly to the rewards $r_e$.

The greedy peeling algorithm by~\citet{zhou2022extracting} operates by iteratively deleting a vertex $v$ with the smallest decrease in $\score$, that is,
\begin{equation}
\label{eq:convpeel}
    \score(S) - \score(S \setminus \set{v}) = \sum_{e: v \in e} r_e(|e \cap S|) - r_e(|e \cap S| - 1),
\end{equation}
and then returning the best tested subgraph.

Unfortunately, this approach yields a guarantee only when $r_e$ is convex and can fail for non-convex rewards.
We extend the approach by replacing the second $r_e$ in Eq.~\ref{eq:convpeel} with a different function $s_e$, which needs to be specified separately. The pseudo-code for the algorithm is given in Algorithm~\ref{algo:peeling}. We will show that certain choices for $s_e$ lead to a guarantee.

\begin{algorithm}[t]
\caption{Approximates the densest subgraph problem, \gdsh}
\label{algo:peeling}
\begin{algorithmic}[1]
%\Require graph $G$. 
\Require Hypergraph $\mathcal{H}=(V, E, r_e)$, and a bound function $s_e(\cdot)$.
\State $X \leftarrow V$ and $Y \leftarrow V$
\While {$X \neq \emptyset$}
	\State $v \leftarrow \arg \min_{u \in X} \sum_{e: u \in e} r_e(|e \cap X|) - s_e(|e \cap X| - 1)$.
    \State $X \leftarrow X \setminus \{v\}$.
	\If{$\dens(X) \geq \dens(Y)$,} $Y \leftarrow X$.
	\EndIf
\EndWhile
\Ensure $Y$.
\end{algorithmic}
\end{algorithm}

The following result shows the conditions required for $s_e$ so that Algorithm~\ref{algo:peeling} yields an approximation guarantee.

\begin{theorem}
\label{thr:peel}
Assume a hypergraph ${H} = (V, E, \{r_e\})$ and a function $s_e$ for each $e \in E$ satisfying $0 \leq s_e(i) \leq r_e(i)$ and $r_e(i) - s_e(i - 1) \leq r_e(i + 1) - s_e(i)$
Then Algorithm~\ref{algo:peeling} yields a $1/k$ approximation for \gdsh.
\end{theorem}

To prove the claim, we need the following standard lemma.
\begin{lemma}
\label{lem:intermed}
Let $O \subseteq V$ be an optimal solution. Then for any $o \in O$, 
\[
\dens(O) \leq \sum_{e: o \in e} r_e(|e \cap O|) - r_e(|e \cap O| - 1).
\]
\end{lemma}    
\begin{proof}
By optimality of $O$, it holds
\[
\dens(O) \geq \frac{\score(O \setminus \{o\})}{|O|-1} = \frac{\score(O)-(\score(O)-\score(O \setminus \{o\})}{|O|-1}. 
\]
Rewriting this inequality and using $\dens(O) = \frac{\score(O)}{|O|}$ shows that $\dens(O) \leq \score(O)-\score(O \setminus \{o\})$. See also \cite[Theorem 4]{zhou2022extracting}.
Thus,
\begin{align*}
\dens(O) &\leq \score(O)-\score(O \setminus \{o\}) \\
&= \sum_{e \in E} (r_e(|e \cap O|)-r_e(|e \cap (O \setminus \{o\})|) \\
&= \sum_{e: o \in e} (r_e(|e \cap O|)-r_e(|e \cap (O \setminus \{o\})|), % \\
\end{align*}
proving the lemma.\qed
\end{proof}

\begin{proof}[Proof of Theorem~\ref{thr:peel}]
Let $O \subseteq V$ be an optimal solution.
Let $X'$ be the subgraph $X$ defined by the while loop in Algorithm~\ref{algo:peeling} when the first vertex from $O$ is deleted. Call that vertex $x$.
Note that since each edge $e$ contains at most $k$ vertices, it holds that
\begin{align*}
k \dens(X') &\geq \frac{1}{{|X'|}}\sum_{u \in X'} \sum_{e: u \in e} r_e(|e \cap X'|) \\
&\geq \sum_{e: x \in e} r_e(|e \cap X'|) \\
& \geq \sum_{e: x \in e} r_e(|e \cap X'|) - s_e(|e \cap X'| - 1),
\end{align*}
where the second inequality follows from our choice of $x$ in Algorithm~\ref{algo:peeling}.
Furthermore, as $O \subseteq X'$ by the imposed conditions on $s_e$ it holds that
\begin{align*}
\sum_{e: x \in e} r_e(|e \cap X'|) - s_e(|e \cap X'| - 1)
& \geq \sum_{e: x \in e} r_e(|e \cap O|) - s_e(|e \cap O| - 1) \\
& \geq \sum_{e: x \in e} r_e(|e \cap O|) - r_e(|e \cap O| - 1).
\end{align*}
Since $x \in O$, we can use Lemma~\ref{lem:intermed}, and the theorem follows.
\qed
\end{proof}

Next, we show two options for $s$, both satisfying the conditions in Theorem~\ref{thr:peel}.

\begin{theorem}
\label{thm:max_se}
Assume a hypergraph $H = (V, E, \{r_e\})$.
Let $b_e(i) = 0$ and $u_e(i) = r_e(i+1) - \max_{0\leq j \leq i} \pr{r_e(j+1)-r_e(j)}$. Then $b_e$ and $u_e$ satisfy the conditions for $s_e$ in Theorem~\ref{thr:peel}. Moreover, any $s_e$ that satisfies the conditions in Theorem~\ref{thr:peel} will have $b_e(i) \leq s_e(i) \leq u_e(i)$. If $r_e$ is convex, then $u_e = r_e$.
\end{theorem}

\begin{proof}
The function $b_e$ satisfies the constraints since  $r_e$ is monotonic.

Let us write $M_e(i) = \max_{0 \leq j \leq i} \pr{r_e(j + 1) - r_e(j)}$.
Then 
\[
    u_e(i) = r_e(i + 1) - M_e(i) \leq r_e(i + 1) - (r_e(i + 1) - r_e(i)) = r_e(i)
\]
and
\[
    u_e(i) - u_e(i - 1) = r_e(i + 1) - r_e(i) - (M_e(i) - M_e(i - 1)) \leq r_e(i + 1) - r_e(i),
\]
showing that $u_e$ satisfies the constraints.

To prove the second claim,
assume that $s_e$ satisfies the constraints and assume inductively that $s_e(i - 1) \leq u_e(i - 1)$. If $M_e(i - 1) < M_e(i)$, then
$M_e(i) = r_e(i + 1) - r_e(i)$, and immediately $s_e(i) \leq r_e(i) = u_e(i)$, proving the claim. Otherwise, assume $M_e(i - 1) = M_e(i)$. Then
\begin{align*}
    s_e(i) & \leq s_e(i - 1) + r_e(i + 1) - r_e(i)  \\
    & \leq u_e(i - 1) + r_e(i + 1) - r_e(i) \\
    & = r_e(i + 1) - M_e(i - 1) 
     = r_e(i + 1) - M_e(i) = u_e(i),
\end{align*}
proving the claim.

If $r_e$ is convex, then $M_e(i) = r_e(i + 1) - r_e(i)$ and $u_e(i) = r_e(i)$, proving the last claim.
\qed
\end{proof}

\noindent \textbf{Runtime analysis.}
Finding the node $v$ in Algorithm~\ref{algo:peeling} can be done with a priority queue. Maintaining such structure requires $\bigO{k \cdot \deg(v)}$ updates when removing a single vertex $v$, each taking $\bigO{\log{n}}$ time for a total running time of $\bigO{pk\log{n}}$, where $p = \sum_{e \in E} |e| = \sum_{v \in V} \deg(v)$.

\subsection{Approximations based on projecting to convexity}
Given a set of reward functions $\{r_e \colon e \in E\}$ and a corresponding objective $\max_{S\subseteq V} \Gamma(S)$, another approach to approximating \gdsh{} is to replace each $r_e$ with a nearby convex function $\hat{r}_e$, and maximize a related objective
$\hat{\Gamma}(S) = \frac{1}{|S|} \sum_{e \in E} w(e) \hat{r}_e(|S \cap e|)$. As an initial observation, these objectives differ by at most the maximum ratio between original ($r_e$) and projected ($\hat{r}_e$) rewards. 
\begin{proposition}
\label{prop:hat}
 If $r_e(i) \geq \hat{r}_e(i)$ for every $e \in E$ and every $i \in [|e|] = \{1, \hdots, |e|\}$, then for every $S \subseteq V$ we have $\hat{\Gamma}(S) \le \Gamma(S) \le \rho \cdot \hat{\Gamma}(S)$, 
where $$\rho = \max_{e \in E} \max_{i \in [|e|]} \frac{r_e(i)}{\hat{r}_e(i)}.$$
 \end{proposition}
\begin{proof}
   The first inequality follows from the assumption that $r_e(i) \geq \hat{r}_e(i)$. The definition of $\rho$ implies that for every $e \in E$ and $i \in [|e|]$, we have ${r}_e(i) \leq \rho \cdot \hat{r}_e(i)$, which yields the second inequality $\Gamma(S) \leq \rho \cdot \hat{\Gamma}(S)$.
\end{proof}
 
This approximation is tight in the sense that we can always construct a hypergraph $H$ with a node set $S$ for which $\Gamma(S) = \rho \cdot \hat{\Gamma}(S)$. 
In more detail, consider a pair of rewards functions $r$ and $\hat{r}$ and let $t = \argmax_{i} \frac{r(i)}{\hat{r}(i)}$. Then construct a $k$-uniform hypergraph $H = (V,E)$ with a node set $S$ satisfying $|e \cap S| \in \{0,t\}$ for every $e \in E$. Use the same reward function $r$ for every edge when defining $\Gamma$, and the reward function $\hat{r}$ for every edge when defining $r$. Then 
\begin{align*}
    \Gamma(S) = \frac{\sum_{e \in E} r(|e \cap S|)}{|S|} = \frac{\sum_{e \in E} \rho \cdot \hat{r}(t)}{|S|} = \rho \cdot \hat{\Gamma}(S).
\end{align*}

Given a non-convex nonnegative monotonic reward function $r \colon [0,k] \rightarrow\mathbb{R}^+$, our goal is then to find a \textit{convex} nonnegative monotonic function $\hat{r} \leq r$ such that $\max_i r(i)/\hat{r}(i)$ is small. This can be cast as a small linear program. 
\begin{align*}
    \maximize_{\hat{r}}  \quad& \kappa   \\
     \text{such that}  \quad & r(i)\kappa  \leq 
 \hat{r}(i) \leq r(i)  & i  = 0, 1,2, \hdots, k  \\
      & 2\hat{r}(i)  \leq \hat{r}(i-1) + \hat{r}(i+1)   &i = 2,3,\hdots,k-1 \\
      & \hat{r}(i+1)  \ge \hat{r}(i) & i  =1,2,\hdots,k-1 
\end{align*}
The resulting approximation factor is given by $\rho = 1/\kappa$. We have dropped the $e$ from the subscript of $r_e$ and $\hat{r}_e$ above since we must solve this generic optimization problem for each edge reward function individually. This problem is equivalent to finding the lower convex hull of the points $\{(0,0), (1, r(1)), \hdots, (k, r(k)) \}$, which can be done in $O(k)$ time~\cite{andrew1979another}. Using the monotonicity of $r$, we can bound the worst-case approximation factor. See Appendix for a proof.
\begin{proposition}
\label{eq:propk}
    Let $r \colon [0,k]  \rightarrow \mathbb{R}^+$ be a monotonic reward function satisfying $r(0) = 0$. There exists a nonnegative monotonic convex function $\hat{r}$ satisfying $\hat{r}(i) \leq r(i) \leq k \cdot \hat{r}(i)$ for every $i \in \{1,2, \hdots, k\}$.
\end{proposition}
Observe that this approximation is tight. If $r$ is defined by $r(0) = 0$ and $r(x) = 1$ for $x \in (0,k]$, then $\hat{r}(x) = x/k$ and $r(1)/\hat{r}(1) = k$. 
Propositions~\ref{prop:hat} and~\ref{eq:propk} tell us that after performing optimal projections, $\Gamma$ and $\hat{\Gamma}$ differ by at most a factor $1/k$, which leads to the following result.
\begin{theorem}
    A $\beta$-approximate solution to $\max_{S \subseteq V} \hat{\Gamma}(S)$ yields a $\frac{\beta}{k}$-approximate solution for $\max_{S \subseteq V} {\Gamma}(S)$.
\end{theorem}
Since $\hat{\Gamma}$ includes only convex edge rewards, we can optimally solve it using flow-based methods, yielding a $1/k$-approximation for the original objective $\Gamma$. To provide a runtime analysis, we assume all original rewards $r_e(i)$ are integers. After projecting, the new rewards $\hat{r}_e$ are not necessarily integers. However, new rewards can be expressed as rational numbers with denominators that range between 1 and $k$. For a simple runtime analysis, we can scale all new rewards by $k!$ to obtain new convex integer reward functions $r_e' = k! \cdot \hat{r}_e \leq k^k \cdot r_e$. The flow-based approach of \citet{zhou2022extracting} for this integer convex rewards case relies on performing a binary search over the interval $[0, W]$ where $W = \sum_{e \in E} r'_e(|e|)\leq k^k \sum_{e \in E} r_e(|e|)$. This has a runtime of $\bigO{\text{mincut}(p,p\cdot k) \log  W}$ time where $p = \sum_{e \in E} |e|$ is the size of the hypergraph and $\text{mincut}(N,M)$ is the complexity of solving a minimum $s$-$t$ cut problem in a graph with $N$ nodes and $M$ edges. To put this expression into a form that only involves the original rewards $\{r_e \colon e \in E\}$, observe that $\log W \leq k \log k + \log (\sum_{e \in E} r_e(|e|))$.

Note finally that if we project the non-convex rewards and apply the existing $1/k$-approximation greedy peeling algorithm for $\hat{\Gamma}$~\cite{zhou2022extracting}, this is only guaranteed to produce a $1/k^2$-approximate solution using this analysis. This again highlights utility of our peeling algorithms that work directly on the non-convex objective and yield a $1/k$-approximate solution.
 \section{Approximation algorithms for constrained variants}
The approximability of \gdsh{} has immediate implications for the approximability of \card{} and \fair{}. We will now explore these results.

Let us first consider the \card{} problem, where the solution must have at least $\ell$ nodes. Here we will adopt the algorithm by~\citet{khuller2009finding} which was used to solve the constrained variant in regular graphs, and further extended to work with supermodular rewards by~\citet{chekuri2022densest}.

The algorithm, given in Algorithm~\ref{algo:card}, iteratively finds an approximate densest \gdsh{}, say $S_i$ from the current hypergraph, say $H_i$, removes $S_i$ from $H_i$ (while keeping the edges), and adds $S_i$ to the solution $S$, until $S$ is large enough. Then the returned value is either $S$, or $S'$, a set corresponding to $S$ during the previous round, plus padded arbitrary nodes to satisfy the constraint.

Algorithm~\ref{algo:card} requires a subroutine for contracting the discovered set from the current hypergraph.
Given a hypergraph $H = (V, E, w, \{r_e\})$ and a set of nodes $U$ we define a contracted hypergraph $H' = (V', E', w', \{r'_e\}) = \contract(H, U)$ as follows.
The nodes are $V' = V \setminus U$, the hyperedges $E'$ consist of the hyperedges in $E$ with nodes in $U$ removed, and the weights $w'$ correspond to the weights $w$. To define the rewards, let $e \in E$ be a hyperedge and $a = e \setminus U$ be the contracted hyperedge. Let $j = \abs{e \cap U}$. We define the contracted reward as $r'_a(i) = r_e(i + j) - r_e(j)$.

\begin{algorithm}[t]
\caption{Approximates the \card{} problem}
\label{algo:card}
\begin{algorithmic}[1]
%\Require graph $G$. 
\Require Hypergraph $H$, the cardinality constraint $\ell$.
\State $S \define \emptyset$, $H_1 \define H$, $i \define 1$
\While {$\abs{S} < \ell$}
    \State $S_i \define $ (approximate) densest \gdsh{} in $H_i$
    \State $H_{i + 1} \define \contract(H_i, S_i)$
    \State $S \define S \cup S_i$
    \State $i \define i + 1$
\EndWhile
\State $S' \define S_1 \cup \cdots \cup S_{i - 2}$ padded with arbitrary nodes so that $S'$ has $\ell$ nodes
\Ensure either $S$ or $S'$, whichever has the higher density.
\end{algorithmic}
\end{algorithm}

We have the following approximation result, which we prove in Appendix.

\begin{theorem}
\label{thr:card}
Assume that we can $\alpha$-approximate \gdsh, then Algorithm~\ref{algo:card} yields $\frac{\alpha}{\alpha + 1}$ approximation for \card{}. Consequently, using Algorithm~\ref{algo:peeling} together with Algorithm~\ref{algo:card} yields $\frac{1}{k + 1}$ approximation.
\end{theorem}

We can now use Theorem~\ref{thr:card} and the algorithm proposed by~\citet{miyauchi2023densest} to obtain an approximation result for \fair.

\begin{theorem}
\label{thr:fair}
Assume an instance \fair{} with $\set{\ell_i}$ constraints. Let $S$ be the $\alpha$-approximation for \card{} with $\ell = \sum_i \ell_i$.
Let $c_i$ be the number of nodes with color $i$ in $S$.
Let $S'$ be $S$ padded with any $\ell_i - c_i$ nodes of color $i$, for every color $i$. Then $S'$ yields $\alpha/2$-approximation for \fair{}. Consequently, using Algorithms~\ref{algo:peeling}~and~\ref{algo:card} yields $1/(2k + 2)$ approximation.
\end{theorem}
Note that originally this algorithm was designed for standard graphs (i.e., hypergraphs with $k = 2$). However, the proof for Theorem~\ref{thr:fair} is identical to the proof by~\citet{miyauchi2023densest}, and therefore omitted. We conjecture that a better approximation is possible using an approach for normal graphs by~\citet{gajewar2012multi}. We leave exploring this direction as a future work.
 \section{Experiments and Analysis}
We now analyze the performance of algorithms over a variety of hypergraphs using several different convex and non-convex reward functions. We implement all the algorithms in Julia and use publicly available hypergraph datasets. All experiments were conducted on a research server with 1TB of RAM.\!\footnote{The code and Appendix are available at repository: \href{https://github.com/Vedangi/Densest-SWAMP}{The Densest SWAMP Problem}.}

\noindent \textbf{Datasets}. The \textit{contact-high-school} (CHS)~\cite{chodrow2021hypergraph,Mastrandrea-2015-contact} and \textit{contact-primary-school} (CPS)~\cite{chodrow2021hypergraph,Stehl-2011-contact} datasets represent student interactions at a high school and primary school, respectively, with students as nodes and group interactions as hyperedges. \textit{Senate-committees} (SC) and \textit{House-committees} (HC) contain labeled nodes representing US Senate and House members with political party affiliations~\cite{chodrow2021hypergraph}. Here hyperedges denote the committee memberships.
In the \textit{Trivago} hypergraph (Triv), nodes are vacation rentals and hyperedges are rentals clicked during the same user browsing session on \texttt{Trivago.com}. We specifically use a subset of a larger hypergraph~\cite{chodrow2021hypergraph}, defined by considering only vacation rentals in Fukuoka, Japan. We preprocess each hypergraph by eliminating self-loops and dangling nodes, and extracting their largest connected component while preserving multi-edges. 

Hypergraph statistics are shown in Table~\ref{tab:peeling-comparison}. We choose hypergraphs that are small enough so that we can find optimal solutions using the ILP, as a point of comparison for our approximation algorithms.
{
\setlength{\tabcolsep}{4pt}
\begin{table}[t!]
    \centering
        \caption{Edge-reward functions and their definitions}
    \label{tab:reward-functions}
    % First subtable
    \begin{subtable}[t]{0.45\linewidth}
        \centering
        % \caption{Part A of Reward Functions}
        \begin{tabular}{cll}
            \toprule
            &\textbf{Function} & \boldmath$r_e$ \\
            \midrule
            1.& atleast-two & 
            \(
              r_e(i) = \mathbf{1}[\,i \ge 2\,]
            \) \\
            2.& atleast-half & 
            \(
              r_e(i) = \mathbf{1}\bigl[\,i \ge \lceil |e|/2 \rceil\,\bigr]
            \) \\
            3.& all-but-one & 
            \(
              r_e(i) = \mathbf{1}[\,i \ge |e|-1\,]
            \) \\
            \bottomrule
        \end{tabular}
    \end{subtable}
    \hspace{0.05\linewidth} % Horizontal space between subtables
    % Second subtable
    \begin{subtable}[t]{0.45\linewidth}
        \centering
        % \caption{Part B of Reward Functions}
        \begin{tabular}{cll}
            \toprule
            &\textbf{Function} & \boldmath$r_e$ \\
            \midrule
            4.& standard & 
            \(
              r_e(i) = \mathbf{1}[\,i = |e|\,]
            \) \\
            5.& quadratic &
            \(
              r_e(i) = i^2/|e|
            \) \\
           6.& square-root &
            \(
              r_e(i) = \sqrt{i}
            \) \\
            \bottomrule
        \end{tabular}
    \end{subtable}
\end{table}
}

\noindent \textbf{Reward functions.} We use a range of edge-reward functions from non-convex to convex. For simplicity, we assume that in a given hypergraph, all edges use the same function to compute $r_e$ and that each edge has weight $1$. The function definitions are presented in Table~\ref{tab:reward-functions}. To avoid trivial solutions (as noted in Theorem~\ref{thr:corner}), we set $r_e(1) = 0$ for functions 1, 2, 3, and 6. This is especially important for 2-node hyperedges while using reward functions 1, 2, and 3.

\noindent\textbf{Algorithms.} 
For finding the optimal solution, we implement the $\exact$ method that iteratively solves the ILP (using Gurobi optimization software) as described in Section~\ref{sec:ilp}.
Instead of using a binary search, we begin with the entire hypergraph and iteratively search for a denser and denser subset until no more improvement is possible, as this tends to converge in 4-5 iterations. We compare $\exact$ against several approximation algorithms.
Algorithm~\ref{algo:peeling} is referred to as as  $\peelmax$ when we set $s_e = u_e$ as specified in Theorem~\ref{thm:max_se}, as $\greedypeel$ when $s_e(i) = r_e(i)$,  and as $\peelzero$ when $s_e(i) = 0$. $\degreepeel$ is the greedy peeling algorithm for the standard densest subhypergraph objective: peeling based solely on the degree (number of incident hyperedges) in the induced hypergraph. For the projection-based approximations, we first project each non-convex $r_e$ onto its nearest convex $\hat{r}_e$. We then run a flow-based method and $\greedypeel$ methods using projected rewards. The projection-plus-flow approach solves the projected problem exactly using a maxflow solver, as described in~\cite{zhou2022extracting}. It comes with a $1/k$-approximation guarantee, while the latter is a $1/k^2$ approximation.

\begin{table}[htbp]
  \centering
  \caption{Objective and runtime values of ILP, peeling and projection based strategies. 
  Dashes indicate that \exact{} did not complete within the allotted time.}
  \label{tab:peeling-comparison}
\begin{tabular*}{\textwidth}{l @{\extracolsep{\fill}} r r r r r r r r r r}
 \toprule
   & \multicolumn{2}{l}{\textbf{CHS}} 
   & \multicolumn{2}{l}{\textbf{CPS}} 
   & \multicolumn{2}{l}{\textbf{SC}}  
   & \multicolumn{2}{l}{\textbf{HC}}  
   & \multicolumn{2}{l}{\textbf{Triv}} \\

   %$|V| \times |E|$ & \multicolumn{2}{c}{327 $\times$ 7818} & \multicolumn{2}{c}{242 $\times$ 12704}  & \multicolumn{2}{c}{282 $\times$ 315}  & \multicolumn{2}{c}{1290 $\times$ 340}  & \multicolumn{2}{c}{262 $\times$ 910} \\
   $|V|, |E|, k$ & \multicolumn{2}{l}{327, 7818, 5} & \multicolumn{2}{l}{242, 12704, 5}  & \multicolumn{2}{c}{282, 315, 31}  & \multicolumn{2}{c}{1290, 340, 81}  & \multicolumn{2}{c}{262, 910, 16} \\
   % $|E|$ & \multicolumn{2}{c}{7818} & \multicolumn{2}{c}{12704} & \multicolumn{2}{c}{315} & \multicolumn{2}{c}{340} & \multicolumn{2}{c}{910} \\
   %$k$ & \multicolumn{2}{c}{5} & \multicolumn{2}{c}{5} & \multicolumn{2}{c}{31} & \multicolumn{2}{c}{81} & \multicolumn{2}{c}{16} \\
   %\hline
   \cmidrule(r){2-3}
   \cmidrule(r){4-5}
   \cmidrule(r){6-7}
   \cmidrule(r){8-9}
   \cmidrule(r){10-11}
   & Obj. & Run & Obj. & Run & Obj. & Run & Obj. & Run & Obj. & Run \\

 \hline
  \multicolumn{11}{c}{\textit{atleast-two}} \\
  \hline
  \exact & 27.078 & 416 & 60.549 & 1760 & -- & -- & -- & -- & 11.53 & 2.4 \\
  \peelmax & 26.871 & 0.15 & 60.02 & 0.42 & 21.0 & 0.89 & 14.0 & 24.01 & 11.09 & 0.06 \\
  \greedypeel & 27.065 & 0.2 & 60.549 & 0.59 & 15.90 & 1.14 & 11.88 & 32.77 & 11.41 & 0.02 \\
  \peelzero & 26.871  & 0.16 & 60.022 & 0.45 & 21.0 & 1.11 & 14.0 & 23.98 & 11.09 & 0.02 \\
  \degreepeel &  26.725& 0.06  & 58.714 & 0.06 & 11.5 & 0.01 & 7.625 & 3.20 & 10.5 & 0.008 \\
  % \projexact & 0.976 & 139.42  & 0.963 & 1059.0 & X & X & X & X & 10.764 & 5.308 \\
  \projflow & 26.451  & 0.16  & 58.34 & 0.21 & 12.22 & 0.07 & 7.7 & 0.43 & 10.76 & 0.02 \\
% \projmax & 0.977 & 3.02 & 0.96 & 4.35 & 11.09 & 1.09 & 7.4 & 23.5 & 10.766 & 0.40 \\
\projgreedy & 26.451 & 0.2  & 58.34 & 0.58 & 12.22 & 0.86 & 7.625 & 25.57 & 10.76 & 0.02\\
 \hline
  % \hline
  \multicolumn{11}{c}{\textit{atleast-half}} \\
  \hline
  \exact & 27.078 & 515  & 60.549 & 1711 & -- &  -- & -- & -- & 9.625 & 2.33 \\
  \peelmax & 26.871 & 0.16 & 60.02 & 0.43 & 3.0 & 0.89& 1.571 & 23.81  & 9.179 & 0.02 \\
  \greedypeel & 27.065 & 0.2 & 60.549 & 0.6 & 2.775 & 1.02 & 2.0 & 29.07 & 9.56 & 0.02 \\
  \peelzero & 26.871 & 0.16 & 60.022 & 0.43 & 3.0 & 1.10 & 2.38 & 23.76 & 9.51 & 0.02 \\
  \degreepeel & 26.725 & 0.05 & 58.714 & 0.07 & 2.66 & 0.01 & 1.321 & 3.2 & 9.02 & 0.008 \\
  % \projexact & 026.451 & 137.58 & 58.34 & 1070.23 & X & X & X & X & 9.378 & 2.80 \\
  \projflow & 26.451 & 0.18 & 58.34 & 0.21 & 2.32 & 0.07 & 2.28 & 0.39 & 9.378 & 0.03 \\
% \projmax & 26.461 & 3.024 & 58.138 & 4.28 & 2.014 & 1.09 & 1.571 & 23.81 & 9.179 & 0.42 \\
\projgreedy & 26.451 & 0.19 & 58.344 & 0.59 & 2.32 & 0.86 & 2.23 & 25.36 & 9.378 & 0.02 \\
 \hline
   % \hline
   
  \multicolumn{11}{c}{\textit{all-but-one}} \\
  \hline
  \exact & 26.926 & 249.9 & 60.16 & 1716 & 1.83 & 6794 & 1.22 & 100.97 & 7.769 & 1.39 \\
  \peelmax & 26.784 & 0.16 & 59.66 & 0.44 & 1.66 & 0.91& 1.07 & 24.32 & 7.725 & 0.02 \\
  \greedypeel & 26.867 & 0.2 & 60.16 & 0.6 & 1.8 & 0.95 & 1.105 & 25.62 & 7.763 & 0.02 \\
  \peelzero & 26.784 & 0.16 & 59.66 & 0.44 & 1.66 & 1.12 & 1.07 & 24.29 & 7.725 & 0.02 \\
  \degreepeel & 26.593 & 0.05 & 58.47 & 0.08 & 1.66 & 0.01  & 0.936 & 3.23 & 7.375 & 0.008 \\
  % \projexact & 26.424 & 111.3 & 58.26 & 899.9 & 1.233 & X & X & X & 7.34 & 0.85 \\
\projflow & 26.424 & 0.28 & 58.26 & 0.32 & 1.233 & 0.05 & 0.976 & 0.32 & 7.34 & 0.02 \\
% \projmax & 26.43 & 3.03 & 58.166 & 4.32  & 1.235 & 1.13 & 0.976 & 24.80 & 7.41 & 0.40 \\
\projgreedy & 26.424 & 0.19 & 55.266 & 0.6 & 1.233 & 0.88 & 1.0 & 24.31 & 7.36 & 0.02\\
 \hline

  \multicolumn{11}{c}{\textit{standard}} \\
  \hline
  % \exact & 25.597 & 16.087 & 54.475 & 24.92 & 1.176 & 7.605 & 0.823 & 9.574 & 5.52 & 7.0 \\
    \exactflow & 25.597 & 0.24 & 54.475 &0.33  & 1.176 & 0.13 & 0.823 & 0.5 & 5.52 & 0.10 \\

  \greedypeel & 25.575 & 0.15 & 54.475 & 0.4 & 1.176 & 0.88 & 0.77 & 24.87 & 5.52 & 0.02 \\
  % \greedypeel & \textbf{27.065} & X & \textbf{60.55} & X & 15.90 & X & 11.875 & X & \textbf{13.92} & X \\
  \peelzero & 25.575 & 0.15 & 54.475 & 0.43 & 1.176 & 1.13 & 0.77 & 24.82 & 5.52 & 0.02 \\
  \degreepeel & 25.581 & 0.15 & 54.475 & 0.14 & 1.176 & 0.04 & 0.794 & 3.22 & 5.52 & 0.04 \\
 \hline

  \multicolumn{11}{c}{\textit{quadratic}} \\
  \hline
   \exactflow & 71.45 & 0.44 & 145.47 & 0.68 & 26.91 & 0.25 & 12.52 & 1.33 & 31.10 & 0.08 \\

  % \exact & 71.45 & 614.62 & 145.46 & 4057.73 & X & Xx & X & Xxxx & 31.10 & 21.045 \\
  % \peelmax & 26.46 & X & 58.13 & X & 11.09 & X & 7.4 & X & 11.09 & X \\
  \greedypeel & 71.34 & 0.15 & 145.377 & 0.42 & 26.91 & 0.84 & 12.52 & 23.50 & 31.10 & 0.02 \\
  \peelzero & 70.874 & 0.16 & 145.06 & 0.41 & 26.78 & 1.09 & 12.25 & 23.14 & 31.10 & 0.02 \\
  \degreepeel & 69.85 & 0.06 & 143.70 & 0.07 & 25.24 & 0.01 & 11.59 & 3.20 & 29.50 & 0.008 \\
   \hline
   
  \multicolumn{11}{c}{\textit{square-root}} \\
  \hline
  \exact & 41.34 & 495  & 91.7 & 4469 & -- & -- & -- & -- & 17.15 & 17.22  \\
  \peelmax & 41.06 & 0.16 & 91.21 & 0.42  & 29.69 & 0.89 & 20.48 & 24.02 & 16.86 & 0.02 \\
  \greedypeel & 41.14 & 0.2 & 91.72 & 0.58 & 29.26 & 0.87 & 18.38 & 24.99 & 17.14 & 0.02 \\
  \peelzero & 41.05  & 0.15 & 90.83 & 0.42 & 29.69 & 1.58 & 17.67 & 23.23 & 16.27 & 0.02 \\
  \degreepeel & 41.06 & 0.06  & 89.92 & 0.06 & 19.64  & 0.01 & 11.94 & 3.28 & 16.16 &0.008  \\
  % \projexact & 0.976 & 139.42  & 0.963 & 1059.0 & X & X & X & X & 10.764 & 5.308 \\
  \projflow & 40.86  & 0.29  & 89.76 & 0.35 & 21.88 & 0.2  & 13.6 & 0.97 & 16.35 & 0.11  \\
% \projmax & 0.977 & 3.02 & 0.96 & 4.35 & 11.09 & 1.09 & 7.4 & 23.5 & 10.766 & 0.40 \\
\projgreedy & 40.61 & 0.2  & 89.76 & 0.58 & 21.88 &   0.86  & 12.80 &24.76 & 16.35 & 0.02 \\
 \bottomrule
\end{tabular*}
\end{table}

{
\setlength{\tabcolsep}{3pt}
\begin{table}[!ht]
    \centering
    \caption{Evaluating edge composition of optimal densest \gdsh{} solutions using five reward functions on two hypergraphs: CHS and CPS.}
    \label{tab:obj-comparison}
    \begin{tabular}{l cc cc cc cc cc}
        \toprule
        & \multicolumn{2}{l} {\texttt{atleast-two} }&  \multicolumn{2}{l}{\texttt{atleast-half}} &  \multicolumn{2}{l}{\texttt{all-but-one}} &  \multicolumn{2}{c}{$\mathtt{|E(S)|}$} &  \multicolumn{2}{c}{$\mathtt{|S|}$} \\ 
        \cmidrule(r){2-3}
        \cmidrule(r){4-5}
        \cmidrule(r){6-7}
        \cmidrule(r){8-9}
        \cmidrule(r){10-11}
        & CHS &CPS &CHS &CPS &CHS &CPS &CHS &CPS &CHS &CPS\\
      \hline
        \textit{atleast-two} &3791  & 7932 &3791  & 7932 &3762  & 7956  & 3053 & 6052  &140 &131 \\ 
       
        \textit{atleast-half} & 3791 & 7932 &3791 & 7932 &3762  &  7956 &3053 & 6052 &140  &131  \\ 
       
        \textit{all-but-one}& 4055 &8349  &4055  &8349  & 4039 &8303 & 3404 &6574 &150 &138 \\ 
        
        \textit{standard} & 6182 & 11341 &6182  & 11341& 6175 &11337 & 6041 & 10895& 236 &200 \\ 
      
        \textit{quadratic} & 757 & 4206 &757  &4206  &755 &4152 & 651 & 3040& 29 &75   \\ 
        
        \bottomrule
    \end{tabular}
\end{table}
}

\noindent\textbf{Performance analysis.} Table~\ref{tab:peeling-comparison} shows runtimes and objective values for our methods. As a first observation, we see that our projection-based flow method is exceptionally fast. Even the peeling methods, although not optimized for runtime, are still orders of magnitude faster than exactly solving the objective using the ILP solver.  For instance, in the CPS hypergraph with the objective using \textit{atleast-two} reward function for every edge, the ILP approach requires approximately 29 minutes, whereas  $\greedypeel$ and \projflow{} find the densest solution in under a second. Furthermore, our approximation algorithms all produce approximation ratios in practice that are close to 1, showing that these methods far exceed their theoretical guarantees and produce nearly optimal solutions. This illustrates that even for cases where \gdsh{} is \np-hard, peeling methods can provide a fast and accurate approach in practice, comparable to the success of peeling methods for poly-time solvable variants. We note also that for these hypergraphs and reward functions, applying direct peeling methods tends to give a slightly better approximation than projecting to nearby convex rewards. As another interesting observation, the standard greedy peeling method (setting $s_e = r_e$) usually produces the best approximation results among peeling methods, despite the fact that our approximation guarantees do not apply to this approach.

\noindent\textbf{Qualitative comparison.} As a final point of comparison, Table~\ref{tab:obj-comparison} reports qualitative aspects of the dense subsets produced by exactly solving \gdsh{} (using \exact{}) with different reward functions on two hypergraphs (CFS and CPS).
 
The columns \texttt{atleast-two}, \texttt{atleast-half}, and \texttt{all-but-one} list the number of hyperedges (fully or partially contained in $S$) that intersect $S$ in at least two nodes, at least $\ceil{|e|/2}$ nodes, or at least $|e|-1$ nodes, respectively. We also report the 
number of edges completely contained ($|E(S)|$) and the subset size ($|S|$).
 
Each reward function leads to a dense subset with distinct characteristics:
for instance, the standard objective tends to produce larger subgraphs whereas the quadratic objective tends to produce smaller subgraphs.

This provides a simple check that finding the densest \gdsh{} using different reward functions does indeed allow us to capture meaningfully different types of density patterns in practice.

\section{Conclusions and Discussion}
We have presented comprehensive hardness results and new approximation algorithms for a dense subhypergraph objective where rewards are given for partially included edges. Our most significant finding is that peeling algorithms can achieve the same approximation guarantee for \np-hard variants of the problem as they do for poly-time variants that are special cases of the densest supermodular subset problem. This is somewhat surprising given that previous approximations for peeling seem to inherently rely on the supermodularity property. As one interesting observation, our theory does not apply to the standard greedy peeling strategy ($s_e = r_e$), but this strategy still seems to work well in practice. One open direction is to explore whether there are cases where the standard greedy strategy performs poorly, or whether we can prove an approximation for this approach using a different technique. Another direction for future work is to explore hardness of approximation results for \np-hard variants of \gdsh{}. 

\begin{credits}
\subsubsection{\ackname} This research is supported by the Academy of Finland project MALSOME (343045) and by the Helsinki Institute for Information Technology (HIIT). The research is also supported by the Army Research Office (ARO) under Award Number W911NF-24-1-0156. The views and conclusions contained in this document are those of the authors and should not be interpreted as representing the official policies, either expressed or implied, of the Army Research Office or the U.S. Government.

\subsubsection{\discintname}
The authors have no competing interests to declare that are
relevant to the content of this article.
\end{credits}
\bibliographystyle{splncs04nat}
\bibliography{bibs/dsh}

% %%%%%%%%%%%%%%%%%%%%%%%%%%%%%%

%\newpage

\appendix
\section*{Supplementary Material for the Densest SWAMP problem}
\addcontentsline{toc}{section}{Appendix}
\section{NP-hardness of $\gdsh(r)$}

In this section we will prove the following result.

\begin{theorem}
\label{prop:nphard}
Assume $r$ such that for some positive integers $i$ and $j$ we have $r(i)/i > r(1)$ and $2r(j) > r(j-1) + r(j+1)$. 
Then $\ugdsh(r)$ is \np-hard.
\end{theorem}

We can safely assume that $i$ is the smallest index for which $r(i)/i > r(1)$.

Let us write $\delta(a) = r(a) - r(a - 1)$. Define $\theta = r(i) - ir(1)$ and $\eta = \delta(j) - \delta(j+1)$.

Let $R = \max r$ and $\Delta = \max \delta$.

We will prove the hardness by a reduction from the independent set problem in 3-regular graphs.
Assume we are given a 3-regular graph $H = (W, A)$ with $n$ nodes and $m$ edges.

We construct $G = (V, E)$ as follows. The nodes consists of the original nodes $W$ and $\ell = k(m (j - 1) +kbn(i - 1))$ nodes $S$, where
\[
    b = \ceil{\frac{4\Delta - r(1)}{\delta(i) - r(1)}}
    \quad\text{and}\quad
    k = \ceil{\frac{4R}{\eta} + \frac{6R}{\theta}}.
\]

Each subset of $j$ nodes in $S$ is connected with a hyperedge with a weight
\[
    \alpha = \beta \frac{\ell}{{\ell \choose i} r(i) } = \frac{\beta i}{{\ell - 1 \choose i - 1}r(i)},
    \quad\text{where}\quad
    \beta = k \frac{5\delta(j) + \delta(j + 1)}{2} + kb\delta(i).
\]
We will denote this set of hyperedges with $E_1$.
For each $(u, v) \in A$ we
connect $u$ and $v$ with $k(j - 1)$ unique nodes in $S$ with $k$ hyperedges, each of size $j + 1$.
For each $u \in W$ we connect $u$ with $kb(i - 1)$ unique nodes in $S$ with $kb$ hyperedges, each of size $i$.
We will denote this set of hyperedges with $E_2$, and the subset of the latter edges with $E'_2$.

Let us write 
\[
    \score_1(X) = \sum_{e \in E_1} r(\abs{X \cap e})
    \quad\text{and}\quad
    \score_2(X) = \sum_{e \in E_2} r(\abs{X \cap e}).
\]

Let $O$ be the optimal solution. Let $T = O \cap S$, and define $d = \abs{T}$ and $c = \abs{O \cap W}$.

We will prove the result as a sequence of lemmas.

\begin{lemma}[Mediant inequality]
\label{lem:frac}
Assume 4 positive numbers $a$, $b$, $c$, and $d$.
Then 
\[
\frac{a+c}{b+d} \geq \frac{a}{b} \iff \frac{c}{d} \geq \frac{a}{b} \iff \frac{c}{d} \geq \frac{a+c}{b+d}.
\]
\end{lemma}

\begin{proof}
\begin{align*}
\frac{a+c}{b+d} \geq \frac{a}{b} \iff \frac{a+c}{b+d} - \frac{a}{b} \geq 0 \iff \frac{b(a+c)-(b+d)a}{b(b+d)} \geq 0 \\
\iff ba+bc-ba-da \geq 0 \iff bc \geq da \iff \frac{c}{d} \geq \frac{a}{b}.
\end{align*}
Similarly,
\begin{align*}
\frac{c}{d} \geq \frac{a+c}{b+d} \iff \frac{c}{d} - \frac{a+c}{b+d} \geq 0 \iff \frac{(b+d)c-d(a+c)}{b(b+d)} \geq 0 \\
\iff bc+dc-da-dc \geq 0 \iff bc \geq da \iff \frac{c}{d} \geq \frac{a}{b},
\end{align*}
proving the claim.
\qed
\end{proof}

\begin{lemma}
\label{lem:opt}
Let $O$ be the densest subgraph. Let $X \subsetneq O$. Then
\[
    \frac{\score(O) - \score(O - X)}{\abs{X}} \geq \dens(O).
\]
Let $Y$ be such that $Y \cap O = \emptyset$. Then
\[
    \frac{\score(O \cup Y) - \score(O)}{\abs{Y}} \leq \dens(O).
\]
\end{lemma}

\begin{proof}
For the first statement, apply Lemma~\ref{lem:frac} with $a = \score(O-X)$, $b = \abs{O-X}$, $c = \score(O) - \score(O - X)$, and $d = \abs{X}$.
By the optimality of $O$, we have 
\[\frac{a+c}{b+d} = \frac{\score(O)}{\abs{O}} \geq \frac{\score(O-X)}{\abs{O-X}} = \frac{a}{b}.\]
Lemma~\ref{lem:frac} then implies 
\[\frac{\score(O) - \score(O - X)}{\abs{X}} = \frac{c}{d} \geq \frac{a+c}{b+d} = \dens(O).
\]

For the second statement, apply Lemma~\ref{lem:frac} with $a = \score(O \cup Y) - \score(O)$, $b = \abs{Y}$, $c = \score(O)$, and $d = \abs{O}$.
By the optimality of $O$, we have 
\[\frac{c}{d} = \frac{\score(O)}{\abs{O}} \geq \frac{\score(O \cup Y)}{\abs{O}+\abs{Y}} = \frac{a+c}{b+d}.\]
Lemma~\ref{lem:frac} then implies 
\[\dens(O) = \frac{\score(O)}{\abs{O}} = \frac{c}{d} \geq \frac{a}{b} = \frac{\score(O \cup Y) - \score(O)}{\abs{Y}},
\]
proving the claim.
\qed
\end{proof}

\begin{lemma}
Let $v \in O \cap W$. Let $Z \subseteq S$ be the set of nodes connected to $v$ by an edge in $E_2'$. There are at least $k$ nodes shared by $Z$ and $O$.
\label{lem:enough}
\end{lemma}
\begin{proof}
Let $v \in O \cap W$. 
Lemma~\ref{lem:opt} implies
\[
    \score(O) - \score(O - v) \geq \dens(O) \geq \dens(S) \geq \beta \geq k b \delta(i).
\]
Let $x$ be the number of hyperedges $e \in E_2'$ for which $e \cap O = \set{v}$. Since there are $k(3 + b)$ edges adjacent to $v$ we have
\[
    \score(O) - \score(O - v) \leq r(1)x + \Delta (k(3 + b) - x).
\]
By definition of $x$, there are $kb - x$ edges in $E_2'$ that contain at least one vertex in $Z \cap O$. Combining the two inequalities and solving for $kb - x$ leads to
\[
    k b - x \geq \frac{k b \delta(i) - k(3 + b)\Delta}{\Delta - r(1)} + bk = \frac{bk(\delta(i) - r_1) - 3k \Delta}{\Delta - r(1)}  \geq k,
\]
where the last inequality is due to the definition of $b$. The claim follows since the edges in $E_2'$ are disjoint in $S$.
\qed
\end{proof}

\begin{lemma}
\label{lem:bound}
$\score_2(O) \leq R(3 + b)d$.
\end{lemma}

\begin{proof}
Let $Y$ be the edges in $E_2$ containing a vertex in $W \cap O$.
Let $Z$ be the edges in $E_2$ that contain a vertex in $O$ but do not contain a vertex in $W \cap O$. Note that $Y$ and $Z$ are the only edges that contribute to $\score_2(O)$. Let $d_1$ be the number of nodes in $O \cap S$ covered by $Y$, and let $d_2$ be the number of nodes in $O \cap S$ covered by $Z$. Note that $d = d_1 + d_2$.

Assume a vertex $v \in W \cap O$. There are $(3 + b)k$ hyperedges in $Y$ attached to $v$. Lemma~\ref{lem:enough} implies that these edges contain $k$ nodes in $O \cap S$. Since edges in $Y$ are disjoint, $\abs{Y} \leq (3 + b)d_1$.

Any edge in $Z$ will have a vertex in $O \cap S$. Since these edges are disjoint in $S$, we have $\abs{Z} \leq d_2$.

Consequently,
\[
    \score_2(O) \leq R (\abs{Y} + \abs{Z}) \leq R((3 + b)d_1 + d_2) \leq R(3 + b)d,
\]
proving the claim.
\qed
\end{proof}

\begin{lemma}
\label{lem:f1}
Let $X \subseteq S$ be a non-empty set. Then
\[
    \score_1(X) \leq \abs{X} r(1) {\ell - 1 \choose i - 1} + \abs{X} \theta {\abs{X} - 1 \choose i - 1}.
\]
The inequality is tight if $X = S$.
\end{lemma}

\begin{proof}
Assume $v \in X$. Then 
\[
\begin{split}
    \score_1(X) & =  \sum_{e \in E_1} r(\abs{X \cap e}) \\
    & = \abs{X} \sum_{e \in E_1; v \in e} \frac{r(\abs{X \cap e})}{\abs{X \cap e}} \\
    & \leq \abs{X} \sum_{e \in E_1; v \in e} r(1) + \sum_{e \in E_1; v \in e \subseteq X} \theta \\
    & = \abs{X} r(1) {\ell - 1 \choose i - 1} + \abs{X} \theta {\abs{X} - 1 \choose i - 1},
\end{split} 
\]
where the inequality is due to the fact that $r(a) / a \leq r(1)$ for any $a < i$ and the definition of $\theta$. Note that the inequality is tight if $X = S$.
\qed
\end{proof}

\begin{lemma}
$S \subseteq O$.
\end{lemma}

\begin{proof}
Assume otherwise. Let $O' = O \cup S$. Note that Lemma~\ref{lem:enough} guarantees that $T$ is not empty.

To prove the claim, we need to show that $\dens(O') > \dens(O)$.
We will do this by bounding the individual terms. We start with $\score_1$, that can be bounded with
\begin{align*}
    \frac{\score_1(O')}{\abs{O'}} - \frac{\score_1(O)}{\abs{O}} & = \frac{\score_1(S)}{\ell + c} -   \frac{\score_1(T)}{d + c}  \\
    & \geq \frac{\score_1(S)}{\ell + c} -   \frac{\score_1(T)}{d + \frac{d}{\ell}c} & (d \leq \ell)\\
    & = \frac{\ell}{\ell + c} \pr{\frac{\score_1(S)}{\ell} -   \frac{\score_1(T)}{d}} \\ 
    & \geq \frac{1}{2}  \pr{\frac{\score_1(S)}{\ell} -   \frac{\score_1(T)}{d}} & (\ell \geq m \geq c)\\ 
    & \geq \frac{1}{2} \theta \pr{{\ell - 1 \choose i - 1} - {d - 1 \choose i - 1}} . & (\text{Lemma~\ref{lem:f1}}) \\
\end{align*}
Consequently,
\begin{equation}
\label{eq:f1}
    \frac{\alpha \score_1(O')}{\abs{O'}} - \frac{\alpha \score_1(O)}{\abs{O}} \geq 
    \frac{\theta \beta i}{2r(i)}\pr{1 - {d - 1 \choose i - 1} / {\ell - 1 \choose i - 1}} \geq \frac{\theta i}{2r(i)}\beta \frac{\ell - d}{\ell - 1},
\end{equation}
where the second inequality is due to the fact that
\[
    {d - 1 \choose i - 1} / {\ell - 1 \choose i - 1}
    \leq \frac{d - 1}{\ell - 1}.
\]

Next we bound $\score_2$ terms with
\begin{equation}
\label{eq:f2}
\begin{aligned}
    \frac{\score_2(O')}{\abs{O'}} - \frac{\score_2(O)}{\abs{O}} & \geq \frac{\score_2(O)}{\abs{O'}} - \frac{\score_2(O)}{\abs{O}} \\
    & = \frac{\score_2(O)}{d}\pr{\frac{d}{\ell + c} - \frac{d}{d + c}} \\
    & \geq R(3 + b)\pr{\frac{d}{\ell + c} - \frac{d}{d + c}} & \quad(\text{Lemma~\ref{lem:bound}}) \\
    & = R(3 + b) d \frac{d - \ell}{(\ell + c)(d + c)} \\
    & \geq R(3 + b) d \frac{d - \ell}{d \ell} \\
    & = R(3 + b) \frac{d - \ell}{\ell}.
\end{aligned}
\end{equation}
Note that the latter two inequalities rely on the fact that $d \leq \ell$.

To match the coefficients in the bounds, we can use the definition of $k$ and $\beta$ and show that
\begin{align}
\label{eq:beta}
    \beta & \geq kb \delta(i) > kb r(1)
    \geq  kb \frac{r(i)}{i} \geq \frac{6bR r(i)}{\theta i}
    > \frac{2(3 + b)R r(i)}{\theta i}.
\end{align}

Finally, we have
\begin{align*}
    \dens(O') - \dens(O) & = \frac{\alpha \score_1(O')}{\abs{O'}} - \frac{\alpha \score_1(O)}{\abs{O}} + \frac{\score_2(O')}{\abs{O'}} - \frac{\score_2(O)}{\abs{O}} \\
    & \geq \frac{\theta i}{2r(i)}\beta \frac{\ell - d}{\ell - 1} + R(3 + b) \frac{d - \ell}{\ell} & (\text{Eqs.~\ref{eq:f1}--\ref{eq:f2}})\\
    & > \pr{\frac{\theta i}{2r(i)}\beta - R(3 + b)}\frac{\ell - d}{\ell} \\
    & > 0, & (\text{Eq.~\ref{eq:beta}})
\end{align*}
which contradicts the optimality of $O$.
\qed
\end{proof}

\begin{lemma}
\label{lem:ind}
Let $(u, v) \in A$. If $u \in O$, then $v \notin O$.
\end{lemma}

\begin{proof}
Let $a$ be the number of neighbors of $u$ in $H$ that are in $O$. The fact that $S \subseteq O$ and Lemma~\ref{lem:opt} imply
\[
    \beta \leq \dens(S) \leq \dens(O) \leq \score(O) - \score(O - u) = kb \delta(i) + ak \delta(j + 1) + (3 - a)k \delta(j).
\]
Now the definition of $\beta$ results in
\[
    2.5\times \delta(j) + 0.5\times \delta(j + 1) \leq a \delta(j + 1) + (3 - a)\delta(j).
\]
Since $\delta(j) > \delta(j + 1)$, this is only possible when $a = 0$.
\qed
\end{proof}

\begin{proof}[Proof of Theorem~\ref{prop:nphard}]
Lemma~\ref{lem:ind} shows that $O \cap W$ is an independent set.
We prove the claim by showing that $O \cap W$ is the maximum independent set.

Since $S \subseteq O$, we can rewrite $\score_2(O)$ as
\begin{align*}
    \score_2(O) & = c k b r(i) + 3ckr(j) + (n - c) k b r(i - 1) + (m - 3c) kr(j - 1)  \\
    & = c k b \delta(i) + 3 c k \delta(j) + \omega,
\end{align*}
where
$\omega = n k b r(i - 1) + m kr(j - 1)$.

The bound
\begin{align*}
    \frac{\alpha{\ell \choose i}r(i) + \omega}{\ell} & = \beta + \frac{\omega}{\ell} \\
    & \leq \beta + R \\
    & \leq \beta + k\eta/4 \\
    & = k \frac{11\delta(j) + \delta(j + 1)}{4} + kb\delta(i) \\
    & \leq  k b \delta(i) + 3k  \delta(j)
\end{align*}
together with Lemma~\ref{lem:frac} imply that
\[
    \dens(O) = \frac{\alpha{\ell \choose i}r(i) + \omega + c k b \delta(i) + 3ck\delta(j)}{\ell + c}
\]
is maximized when $c$ is maximized.
\qed
\end{proof}

\section{Projection-based Method Approximation Guarantee}
\begin{proposition}
\label{eq:propkapp}
    Let $r \colon [0,k]  \rightarrow \mathbb{R}^+$ be a monotonic reward function satisfying $r(0) = 0$. There exists a nonnegative monotonic convex function $\hat{r}$ satisfying $\hat{r}(i) \leq r(i) \leq k \cdot \hat{r}(i)$ for every $i \in \{1,2, \hdots, k\}$.
\end{proposition}
\begin{proof}
    To prove the result, we explicitly describe a procedure for finding the optimal $\hat{r}$. Although not the most efficient in terms of runtime, the procedure is conceptually simple and allows us to prove the desired bound.  

    At step 0, we initialize the procedure by setting $t_0 = 0$ and $\hat{r}(0) = 0$. At the beginning of step $i$, we assume we have already defined $\hat{r}$ on the interval $[0,t_{i-1}]$ for some positive integer $t_{i-1}$. If $t_{i-1} < k$, we then find a new point $t_{i}$ as
    \begin{align}
    \label{eq:slope}
        t_{i} = \min_{j \in \{t_{i-1} + 1, \hdots, k \}} \frac{r(j) - r(t_{i-1})}{j - t_{i-1}}.
    \end{align}
    In other words, we find the line connecting $(t_{i-1}, r(t_{i-1}))$ to another point $(j, r(j))$ with the smallest possible slope. We then define $\hat{r}$ on the interval $[t_{i-1}, t_{i}]$ to be the line segment from $(t_{i-1}, r(t_{i-1}))$ to $(t_{i}, r(t_{i}))$. This continues until an iteration $m$ where $t_m = k$.

    The proof follows by showing three properties for this procedure:
    \begin{enumerate}[label={(\alph*)}]
        \item $\hat{r}(j) \leq r(j)$ for $j \in \{0,1, \hdots k\},$ 
        \item $\hat{r}$ is convex, and
        \item $r(j) \leq k \cdot \hat{r}(j)$ for $j \in \{0,1, \hdots k\}.$ 
    \end{enumerate}
    Property (a) follows from the fact that at each iteration, we chose the line segment with the smallest possible slope. To prove Property (b), consider the points $\{t_{i-1}, t_i, t_{i+1}\}$ identified in three consecutive iterations of the procedure. 
    Using the minimality property of $t_i$ in Eq.~\eqref{eq:slope}, we have
    \begin{align*}
        \frac{r(t_i) - r(t_{i-1})}{t_i - t_{i-1}} &\leq  \frac{r(t_{i+1}) - r(t_{i-1})}{t_{i+1} - t_{i-1}} = \frac{[r(t_{i+1}) - r(t_i)] + [r(t_i) - r(t_{i-1})]}{[t_{i+1} - t_i] + [t_i -t_{i-1}]}  \\
        &=\alpha \frac{r(t_{i+1}) - r(t_i)}{t_{i+1} -t_i} + (1-\alpha) \frac{r(t_i) - r(t_{i-1})}{t_i -t_{i-1}},
    \end{align*}
    where the last equation must hold for some $\alpha \in [0,1]$. Rearranging gives
    \begin{align*}
        \frac{r(t_i) - r(t_{i-1})}{t_i - t_{i-1}}
        &\leq \frac{r(t_{i+1}) - r(t_i)}{t_{i+1} -t_i}.
    \end{align*}
    In other words, the slopes from iteration $i$ to iteration $i+1$ are nondecreasing. Since $\hat{r}$ is a piecewise linear function with nondecreasing slopes, it is convex.
    
    Finally, to show Property (c), it is sufficient to show that the bound holds for each individual line segment. 
    On interval $[t_i,t_{i+1}]$, $\hat{r}$ is defined by the line
    \begin{align*}
        \hat{r}(x) = \frac{r(t_{i+1})-r(t_i)}{t_{i+1}-t_i}(x - t_i) + r(t_i).
    \end{align*}
    For $j \in \{t_i + 1, \cdots, t_{i+1} - 1\}$, we have
    \begin{align*}
        \frac{r(j)}{\hat{r}(j)} &= \frac{r(j)(t_{i+1} - t_i)}{(r(t_{i+1})-r(t_i))(j-t_i) + r(t_i)(t_{i+1} - t_i)} \\
        &= \frac{r(j)(t_{i+1} - t_i)}{r(t_{i+1})(j-t_i) + r(t_i)(t_{i+1} - j)} \\
        &\leq \frac{r(j)}{r(t_{i+1})} \cdot \frac{(t_{i+1} - t_i)}{(j-t_i)} \leq k.
    \end{align*}
The last inequality follows from the monotonicity of $r$ and the fact that $j - t_i \geq 1$ and $t_{i+1} - t_i \leq k.$   
\qed
\end{proof}

\section{Approximation Guarantee for \card}

\begin{theorem}
\label{thr:card}
Assume that we can $\alpha$-approximate \gdsh, then Algorithm~\ref{algo:card} yields $\frac{\alpha}{\alpha + 1}$ approximation for \card{}. Consequently, using Algorithm~\ref{algo:peeling} together with Algorithm~\ref{algo:card} yields $\frac{1}{k + 1}$ approximation.
\end{theorem}

\begin{proof}
Let us write $\beta = \alpha/(1 + \alpha)$.

Let $O$ be the optimal solution, and let $S_i$, $H_i$, $S'$ and (final) $S$ be as used in Algorithm~\ref{algo:card}.

Let $\dens_i$ be the density function in $H_i$,
and let $T_i$ be the subgraph optimizing the density $\dens_i$.

Define $W_i = \bigcup_{j < i} S_j$ be the nodes removed from $H$ to obtain $H_i$. Let $c$ be the index of the last $S_i$ added, and write $W = W_c$.

We split the proof in two cases:

Assume that $\score(W) \leq \beta\score(O)$. We have immediately
\[
\begin{split}
    \dens_i(T_i) & \geq \dens_i(O \setminus W_i) \\
    & = \frac{\score(O \cup W_i) - \score(W_i)}{\abs{O \setminus W_i}}
    \geq\frac{\score(O) - \score(W)}{\abs{O \setminus W_i}}
    \geq \frac{1}{\alpha + 1}\frac{\score(O)}{\abs{O \setminus W_i}} \geq \frac{\dens(O)}{\alpha + 1} .
\end{split}
\]
Consequently, $\dens_i(S_i) \geq \beta\dens(O)$ and
\begin{align*}
\begin{split}
    \dens(S) = \frac{\sum_i \score(S_i \cup W_i) - \score(W_i)}{\sum_{i} \abs{S_i}}
    = \frac{\sum_i \abs{S_i}\dens_i(S_i)}{\sum_{i} \abs{S_i}}
    \geq \beta  \frac{\sum_i \abs{S_i}\dens(O)}{\sum_{i} \abs{S_i}} = \beta\dens(O),
\end{split}
\end{align*}
proving the first case.

Assume now that $\score(W) \geq \beta \score(O)$.
Then
\[
    \dens(S') = \frac{\score(S')}{\ell} 
    \geq \frac{\score(W)}{\ell} \geq \beta\frac{\score(O)}{\ell} \geq \beta\dens(O),
\]
proving the second case.
\iffalse
Assume that $\score(W \cap O) \leq \beta\score(O)$. We have immediately
\[
    \dens_i(T_i) \geq 
    \frac{\score(O) - \score(W_i \cap O)}{\abs{O \setminus W_i}} 
    \geq\frac{\score(O) - \score(W \cap O)}{\abs{O \setminus W_i}}
    \geq \frac{1}{\alpha + 1}\frac{\score(O)}{\abs{O \setminus W_i}} \geq \frac{\dens(O)}{\alpha + 1} .
\]
Consequently, $\dens_i(S_i) \geq \beta\dens(O)$ and
\begin{align*}
    \dens(S) = \frac{\sum_i \score(S_i) - \score(S_i \cap W_i)}{\sum_{i} \abs{S_i}}
    = \frac{\sum_i \abs{S_i}\dens_i(S_i)}{\sum_{i} \abs{S_i}}
    \geq \beta  \frac{\sum_i \abs{S_i}\dens(O)}{\sum_{i} \abs{S_i}} = \beta\dens(O),
\end{align*}
proving the first case.

Assume now that $\score(W \cap O) \geq \beta \score(O)$.
Then
\[
    \dens(S') = \frac{\score(S')}{\ell} \geq \frac{\score(S' \cap O)}{\ell} \geq \frac{\score(W \cap O)}{\ell} \geq \beta\frac{\score(O)}{\ell} \geq \beta\dens(O),
\]
proving the second case.
\fi
\qed
\end{proof}

%% Note that this preceding line implies that you store your BibTeX references in a file called 'mybibliography.bib'. If you instead store your references in a file with a different name, for instance 'references.bib', the preceding line should read '\bibliography{references}'. Whatever you do, DO NOT put the file name extension .bib inside the \bibliography command; this will trip up LaTeX compilers. 
%
% If you do not want to use BibTeX, you can also type up the bibliography exactly as you see fit, using the following structure:

% Note that this number 8 reserves an amount of space (equal to the natural width of the given number) for the label of your references; if you have more than 9 references, you will want to change this number to 18. If you have more than 19 references, this number is best changed to 88. If you have more than 99 references, I salute you.




% LNCS Homepage, \url{http://www.springer.com/lncs}, last accessed 2023/10/25
% \end{thebibliography}
\end{document}